\documentclass{article}
\usepackage{alphabeta}
\usepackage[utf8x]{inputenc}
\usepackage{amsmath}
\usepackage{hyperref}
\usepackage{amsthm}
\usepackage{float}
\usepackage{graphicx}
\graphicspath{ {./images/} }
\usepackage{wrapfig}
\newtheorem{theorem}{Theorem}
\newtheorem{lemma}{Lemma}
\newtheorem{corollary}{Corollary}
\usepackage{subcaption}
\usepackage{xcolor}
\newcommand{\bey}{\begin{eqnarray}}
\newcommand{\eey}{\end{eqnarray}}

\usepackage{subcaption}
\usepackage[hmargin=3.5cm,vmargin=3.5cm]{geometry}
\usepackage{cite}
 \usepackage{tabularx}
\begin{document}

\title{Singular Solutions of the Tolman–Oppenheimer–Volkoff Equation with a Cosmological Constant: Classification and Properties}
  \author {Christos Dounis\footnote{c.ntounis@ac.upatras.gr } \; and Charis Anastopoulos\footnote{anastop@upatras.gr},  \\
 {\small Laboratory of Universe Sciences, Department of Physics, University of Patras, 26500 Greece} }

\maketitle

\begin{abstract}
 We study the Tolman–Oppenheimer–Volkoff equation in the presence of a cosmological constant   for general thermodynamically consistent equations of state, without imposing regularity at the center. Formulating the problem as an initial value system integrated from an outer boundary inwards, we obtain a general classification of solutions and show that singular configurations dominate the solution space. We demonstrate that all singular solutions share a universal geometric structure and give rise to spacetimes that are bounded-acceleration complete, indicating that the associated singularities are comparatively mild.

Our results extend the classification previously obtained for $\Lambda=0$ [Class. Quant. Grav. 38, 195026 (2021)] and reveal qualitatively new features for $\Lambda \neq 0$. For $\Lambda < 0$, we identify solutions with approximate horizon structures that mimic black holes in equilibrium with their Hawking radiation. For $\Lambda > 0$, we find four distinct classes of solutions with cosmological horizoms, distinguished  by the behavior of their temperature gradients.
\end{abstract}

\section{Introduction}

In this paper, we study the Tolman–Oppenheimer–Volkoff equation in the presence of a nonzero cosmological constant $\Lambda$ (TOV-$\Lambda$) for general equations of state (EoS), without imposing regularity at the center. We formulate the TOV-$\Lambda$ system as an initial value problem, integrating from an outer boundary inwards, and focus on the generic class of solutions that develop singularities. Our aim is to provide a systematic classification of these singular solutions and to analyze their analytic and geometric properties, with particular emphasis on the role of the cosmological constant. We identify structural features common to all singular configurations and characterize their dependence on the sign and magnitude of $\Lambda$, covering both asymptotically de Sitter and anti–de Sitter cases.

This work generalizes the analysis of Ref.~\cite{AnSav21}, where the corresponding classification was carried out for $\Lambda = 0$; for earlier work on singular solutions to the TOV equations see Refs. \cite{ZuPa84, CoKa94, AnSav12, Kotop1, Kim17}. A key condition enabling such a classification is that the EoS is thermodynamically consistent, i.e., derived from an entropy function satisfying the standard principles of thermodynamics. The same condition applies for $\Lambda \neq 0$. In contrast, inconsistent EoS generically lead to the inward integration terminating before the center, typically through a divergence or vanishing of the density, while the local temperature remains finite.

A nonzero cosmological constant modifies the TOV system in a fundamental way: it acts as a repulsive contribution for $\Lambda > 0$ and as a confining mechanism for $\Lambda < 0$. While regular solutions continue to model compact objects, their properties and interpretation depend sensitively on the sign and magnitude of $\Lambda$, and must be reconsidered within either a cosmological or an asymptotically anti–de Sitter setting.

As in the $\Lambda = 0$ case, regular solutions form a set of measure zero within the full space of solutions, and the generic solutions of the TOV-$\Lambda$ system are singular. Despite their abundance, such solutions have received little attention in the literature, likely because they involve naked singularities. In this paper, we focus on their analytic and geometric properties. We deliberately refrain from assigning direct physical interpretation at this stage, aiming instead to establish a systematic mathematical classification; the physical implications will be addressed separately.

The motivation for this work is three-fold. First, the classification of such solutions is of intrinsic interest, especially in view of the qualitatively new features introduced by $\Lambda$, including cosmological or AdS asymptotics, modified causal structure, and the presence of horizons. Existing studies of the TOV-$\Lambda$ equation are relatively limited and concentrate mainly on particular properties of regular solutions, without addressing the global structure of the solution space \cite{Stuch00, Bohmer, BoHa05, HsuMa16}. While singular solutions have been studied in the case $\Lambda=0$, to our knowledge the only analysis for $\Lambda \neq 0$ is Ref.~\cite{Winter}, which identifies some of their basic features.

We show that singular solutions to the TOV-$\Lambda$ equation exhibit universal structural properties, including characteristic geometric behavior near the singularity that is largely independent of the asymptotic regime. Crucially, the singularities are benign: the resulting spacetimes are bounded-acceleration complete, a key requirement for physically admissible singular solutions \cite{HawkingEllis, Geroch}.

Our second motivation arises from the thermodynamics of self-gravitating systems. Previous work has suggested that singular configurations are required for thermodynamic consistency, even when interpreted as auxiliary states \cite{AnSav12, Kotop1, kotop2}. This is achieved through a consistent assignment of entropy to the singularity, in accordance with Penrose's long-standing conjecture \cite{Penrose1, Penrose2}. Extending this framework to $\Lambda \neq 0$ backgrounds is natural: for $\Lambda > 0$, cosmological horizons introduce additional thermodynamic structure \cite{GH77}, while for $\Lambda < 0$, the asymptotically anti–de Sitter setting suggests connections with both holographic entropy \cite{Maldacena1997,RyuTakayanagi2006} and the extended thermodynamic interpretation of asymptotically AdS spacetimes \cite{KRT09, KM12, KMT17}. In the latter context, the cosmological constant is associated with a thermodynamic pressure and the mass with enthalpy. Hence, singular configurations may be viewed as equilibrium states within a broader thermodynamic phase space.

Our third motivation is that singular TOV solutions may provide insight into physically relevant regimes of gravitational dynamics. For $\Lambda > 0$, they can be interpreted as idealized descriptions of highly compressed or rapidly evolving matter distributions in a cosmological background, potentially relevant to early-universe physics. For $\Lambda < 0$, they admit a natural interpretation within the AdS/CFT correspondence, where bulk singularities may encode limiting regimes or phases of the dual strongly coupled field theory. In both cases, it is important to determine whether such configurations can arise as endpoints or approximations of dynamical processes such as gravitational collapse.

The results presented here provide, to our knowledge, the first complete characterization of singular solutions to the TOV-$\Lambda$ system.

The structure of the paper is as follows. In Sec.~2, we set up the notation and define thermodynamic consistency of the EoS. In Sec.~3, we prove a key result: integration of the TOV-$\Lambda$ equations from the boundary inwards encounters no horizon and proceeds up to the center. We also analyze properties of the solutions that are independent of the sign of $\Lambda$, including the characterization of the singularities.

In Sec.~4, we study asymptotically anti–de Sitter solutions and show that their structure closely parallels the $\Lambda = 0$ case. In particular, they admit approximate-horizon configurations that mimic black holes in equilibrium with a thermal atmosphere. In Sec.~5, we analyze asymptotically de Sitter solutions and identify four distinct classes, distinguished by the behavior of their temperature gradients. Finally, Sec.~6 contains a discussion of our results.

\section{Background}
\subsection{Tolman–Oppenheimer–Volkoff equation with a cosmological constant}
We analyze static spherical symmetric solution to Einstein's equation with a cosmological constant $\Lambda$. In the standard coordinate system, they are of the form 
\begin{equation} \label{eq3}
ds^2 = - L(r)^2 dt^2 + \Bigl( 1 - \frac{2m(r)}{r} -\frac{\Lambda}{3}r^2 \Bigr) ^ {-1}   dr^2 +r^2 d \Omega ^2
\end{equation}
where $m(r)$ is the mass function, related to the energy density by 
\begin{equation} \label{eq5}
\frac{dm}{dr} = 4 \pi r^2 \rho
\end{equation}
and $L(r)$ is the lapse function, expressed in terms of the local temperature by  Tolman's equation $LT = $ constant.

We consider solutions for which  the mass function  becomes constant, $m(r) = M$ for  $r > r_B$. If the energy density is continuous at $r = r_B$, then the boundary is a stellar surface. If the energy density is discontinuous, then the surface $r = R$ is interpreted as a bounding box---an idealization useful for the thermodynamic analysis of self-gravitating system.

For $r \leq r_B$, the mass function $m(r)$ and the energy density $\rho$ satisfy the 
 TOV-$\Lambda$ equation,  
\begin{equation} \label{eq4}
\frac{dP}{dr} =  - (\rho + P)   \frac{  m + 4 \pi r^3   P  -  \frac{ \Lambda r^3}{3}  }{r^2 \Bigl( 1 - \frac{2m }{ r}  - \frac{\Lambda r^2}{3}      \Bigr)} 
\end{equation}
where   $P$ is the pressure of matter.  
This set-up is consistent, only if the radius $r_B  > r_0(M, \Lambda)$, $r_0(M,\Lambda)$ denotes the black hole event horizon of the Schwarzschild spacetime with a cosmological constant (Schwarzschild–de Sitter or Schwarzschild–anti–de Sitter) \cite{Kottler,StHl99},
 \begin{equation} \label{eq2}
ds^2 = - \Bigl( 1 - \frac{2M}{r} -\frac{\Lambda}{3}r^2 \Bigr)dt^2 + \Bigl( 1 - \frac{2M}{r} -\frac{\Lambda}{3}r^2 \Bigr) ^ {-1}   dr^2 +r^2 d \Omega ^2.
\end{equation}
The function $r_0(M, \Lambda)$ is determined by the vanishing of $g_{tt}$ in Eq. (\ref{eq2}), which corresponds to the solution of  a third order polynomial,
\bey
r_0(M, \Lambda) = \left\{ \begin{array}{cc} \frac{2}{\sqrt{\Lambda}} \sin\left[\frac{1}{3} \sin^{-1}\left(3M\sqrt{\Lambda}\right)\right], & \Lambda > 0,
\\
\frac{2}{\sqrt{|\Lambda|}} \sinh\left[\frac{1}{3} \sinh^{-1}\left(3M\sqrt{|\Lambda|}\right)\right], & \Lambda < 0. \end{array}\right. \label{limit}
\eey
For $\Lambda > 0$, $r_B < r_1(M, \Lambda)$, where $r_1(M, \Lambda)$ is the exterior (cosmological) horizon of the S-dS solution,
\bey
r_1(M, \Lambda) = \frac{2}{\sqrt{|\Lambda|}} \sin\left[\frac{\pi}{3}-\frac{1}{3} \sin^{-1}\left(3M\sqrt{|\Lambda|}\right)\right].
\eey
These equations imply an upper limit to the mass: $M < \frac{1}{3\sqrt{\Lambda}}$, otherwise the S-dS spacetime has no static region.

\subsection{Thermodynamic consistency}

The TOV-$\Lambda$ equations can be solved once an equation of state (EoS) specified, that is,
a relation between $P$ and $\rho$. However, not all such relations are physically
meaningful. In Ref. \cite{AnSav21}, it was argued that thermodynamic consistency is an
essential physical criterion for an EoS.

This means that the EoS must be generated by a consistent entropy functional:  an entropy density $s(\rho, n_a)$ that is a local function of the energy
density $\rho$ and the particle densities $n_a$; the index $a$ labeling different particle
species. Entropy maximization in equilibrium implies that the local temperature $T^{-1} = \partial s/\partial \rho$
 satisfies Tolman’s law: the product $\sqrt{-g_{00}}T $ is constant, and
that the fugacities $b_a = - \partial s/\partial n_a$ are also constant \cite{KM75, SavAn14}.

This implies that there is a natural thermodynamic representation for   self-gravitating systems
based on the free entropy density $\omega(\rho, b_a)$, This is defined as the Massieu function, obtained from the Legendre transform
\bey
\omega = s - \sum_a  \frac{\partial s}{\partial n_a}  n_a = s +\sum_a b_a n_a
\eey
of the entropy density $s$ with respect to the number densities $n_a$.

In this representation, the temperature is given by $T^{-1} = \partial \omega/\partial \rho$, and the pressure by $P = \omega T - \rho$. These relations enable us to express the 
 density $\rho(T, b_a)$ and the pressure $P(T, b_a)$,
as   functions of the local temperature $T$ and the fugacities $b_a$. Since the
latter are constant for any solutio, we will be dropping the dependence of pressure and density
on $b_a$, and write their functional form as $\rho(T)$ and $P(T)$.
Then, the Gibbs-Duhem relation becomes
\bey
\frac{dP}{dT} = \frac{\rho + P}{T}. \label{gibbsd}
\eey
Thermodynamic stability implies that $(\partial^2 \omega/\partial \rho^2)_{b_a} < 0$, hence that $\partial \rho/\partial T > 0$ at all $T$. 

An EoS $P = f(\rho)$ is thermodynamically consistent only when Eq. (\ref{gibbsd}) can be
integrated for all $T$, that is, if the vector field 
\bey
X_{f} \equiv \frac{x+f(x)}{f^ {'} (x)} \frac{\partial}{\partial x}
\eey
on $\textbf{R} ^{+} _{*}$ is complete.
An incomplete vector field $X_f$ for the EoS can result in  infinite, zero, or  negative 
density $\rho$ at a finite temperature $T$. When used in the TOV-$\Lambda$ equation, such
an EoS typically leads to divergences of physical quantitis at finite values of r.

We also impose the following constraints on $\rho$ and $P$.
\begin{itemize}
\item $\rho \geq 0$ (weak energy condition);
\item  $P \geq  0$ (dynamical stability);
\item $P \leq \rho$ (dominant energy condition).
\end{itemize}
It is straightforward to show that these conditions are compatible with Eq. (\ref{gibbsd}),
only if $P(0) = \rho(0) = 0$. Furthermore, 
\bey
\frac{d \ln P}{d \ln T} \geq 2. \label{exponent}
\eey
This implies that $P$ and $\rho$ grow at least as $C(b_a) T^2$ for $T$ near zero and as $T\rightarrow \infty$, where $C(b_a)$ is a function of the fugacities.

The thermodynamic identity $N_a = T^{-1}(\partial P/\partial b_a)_T$ implies that near $T = 0$, $N_a < C'(b_a) T$ and we conclude that $N_a(0) = 0$. The zero temperature limit for fixed fugacity corresponds essentially to a dilute gas. Note that the degenerate Fermi gas limit corresponds to $T \rightarrow 0$ with fixed particle densities $n_a$, not fixed fugacities.


There are no other purely thermodynamic constraints on possible EoS. There are, however, constraints from statistical mechanics and 
relativistic quantum field theory (QFT). 

\begin{itemize}
\item In a theory with a mass gap $\bar{m}$,  the $T \rightarrow 0$ regime describes a dilute gas of massive particles as $T\rightarrow 0$ , whence,  $ \rho(T) \sim T^{5/2} e^{-\bar{m}/T}$.
\item In a theory with no mass gap, $\rho(T) \simeq T^4$ near $T = 0$.
\item In renormalizable theories, the asymptotic EoS as $T \rightarrow \infty$ corresponds to a gas of massive particles, $\rho(T) \rightarrow T^4$, modulo logarithmic corrections from coupling constant renormalization \cite{Collins}. Even for non-renormalizable theories $\rho(T)$ grows at most polynomially with $T$.
\end{itemize}


\subsection{Alternative forms of the  TOV-$\Lambda$ equation}

By Eq. (\ref{gibbsd}), the TOV-$\Lambda$ equation becomes
\bey
\frac{d\ln T}{dr} =  -   \frac{  m + 4 \pi r^3   P  -  \frac{ \Lambda r^3}{3}  }{r^2 \Bigl( 1 - \frac{2m }{ r}  - \frac{\Lambda r^2}{3}      \Bigr)}. \label{tovt}
\eey

It is convenient to introduce the radial coordinate $\xi : = -\log \left(\frac{r}{r_{B}}\right)$ with values in $[0, \infty)$, and the variables

\bey  
u : = \frac{2m}{r} + \frac{\Lambda}{3} r^2, \hspace{0.3cm} v : = 4 \pi r^2 \rho, \hspace{0.3cm} w : = 4 \pi r^2 P   \leq v,  \hspace{0.3cm}
t : = \log (\frac{T}{T_{B}}),
\label{uvw}
 \eey
 where $T_{B}$ denotes the temperature at the stellar surface. Hence, Eqs. \eqref{eq5}, \eqref{eq4}  become 

\bey
t' =   \frac{  \frac{u}{2} + w - \frac{\lambda}{2}   e^{-2 \xi}     }{1 - u},
\;\;\; \;\;\;
u' =  u - 2v - \lambda   e^{-2 \xi}, 
\label{eqqq}
\eey
where $\lambda = \Lambda r_{B}^2$, and prime denotes differentiation with respect to $\xi$.

An alternative parameterization is obtained by defining
\bey
\tilde{P} = P - \frac{\Lambda}{8 \pi}, \; \;  \tilde{\rho} = \rho + \frac{\Lambda}{8\pi}, \;\; \tilde{m} = m + \frac{ \Lambda r^3}{6}.
\label{tilded}
\eey
With these new variables the TOV-$\Lambda$ coincides with the standard TOV equation (without a cosmological constant), the difference being that the EoS $P = f(\rho)$ is substituted with the EOS $\tilde{P} = \tilde{f}(\tilde{\rho})$, where $\tilde{f}(x) = f(x - \frac{\Lambda}{8 \pi}) - \frac{\Lambda}{8 \pi}$. The key point here is that if $f(x)$ satisfies the integrability condition for thermodynamic consistency, so does $\tilde{f}(x)$. The only difference lies in the initial condition, we have  
$\tilde{P}(0) = - \frac{\Lambda}{8 \pi}$ and $\tilde{\rho}(0) = \frac{\Lambda}{8\pi}$.

By introducing the variables 
\bey
 \tilde{v} : = 4 \pi r^2 \tilde{\rho}, \;\; \tilde{w} : = 4 \pi r^2  \tilde{P},
\eey
we write Eq. (\ref{eqqq}) as
\bey 
t' =   \frac{\frac{u }{2} + \tilde{w}    }{1 - u}, \;\; \;  u' =  u - 2 \tilde{v}.  \label{seteq}
\eey

\section{Structure of solutions to the TOV-$\Lambda$ equation}
In this section, we prove the following theorem.
\begin{theorem}
Integration of the TOV-Λ equations from the boundary inwards for a thermodynamically consistent EoS continues all the way to the center. Two types of solutions arise:  regular ones $(m(0) = 0)$ and singular ones with $\lim_{r\to 0} m(r) < 0$. Singular solutions have finite $m(0) < 0$ and temperature that vanishes with $\sqrt{r} $ as $r \to 0$.
\end{theorem}

The first step is to prove that integration from the boundary inwards does not encounter a horizon, i.e., that there is no $\xi > 0$, such that $\tilde{u}(\xi) = 0$. To prove this for all values of the cosmological constant, we need two lemmas, each  covering a different range of $\lambda$.

\begin{lemma}
If $\lambda \leq 1$,   then $u(\xi)<1$   for all $\xi \ge 0 $.
\end{lemma}

\begin{proof}
We start integrating Eq. (\ref{eqqq}) from $\xi = 0$, where $u < 1$. Suppose that $u(\xi)$ first becomes equal to unity at $\xi = \xi_*$. For $\xi < \xi _{*}$,  $u$ is at least a $C ^{1}$ function of $\xi$. Defining $\epsilon := 1-u >0$ and $x := \xi^* - \xi > 0$, we find that in a small neighbourhood of $\xi = \xi_*$, 
 \bey
 \frac{d \epsilon}{d x} =  1 - 2v_* - \lambda  e^{ -2 \xi _{*}}\\
 \frac{dt}{dx} = -  \frac{  \frac{1}{2} + w_* - \frac{\lambda}{2}  e^{-2 \xi _{*} }     }{\epsilon},
 \eey
 where $v_* = v(\xi_*)$, and $w_* = w(\xi_*)$. 
 
To reach $\epsilon = 0 $,  $ \frac{d \epsilon}{d x} \ge 0  $ as $x \to 0 ^{+} $. Hence, $v_*  \le \frac{1}{2} - \frac{\lambda }{2} e^{-2 \xi _{*}} $. When the inequality is sharp,   $\epsilon = (1 - 2v_* - \lambda  e^{ -2 \xi _{*}})x$.

Since $w_* < v_*$, 
\bey
- \frac{dt}{dx}  \le \frac{1- \lambda    e^{-2 \xi _{*} }} {\epsilon} = \frac{C}{x}, 
 \eey
where $C = (1 - \lambda  e^{-2 \xi _{*}})/ (1 - 2v_* - \lambda  e^{ -2 \xi _{*}})$. For $\lambda < 1$, $C > 0$. Hence, integrating from a reference point $x_r$ to $x$, we obtain
\begin{equation} 
\ln \frac{T(x)}{T(x_{r})} \ge C \ln \Bigl( \frac{x_{r}}{x} \Bigr) 
\end{equation}
Hence, $T(\xi_*) = \infty$, leading to $\rho(\xi_*) = \infty$, in contradiction to the upper bound in $v_*$ established earlier. 

If $v_* = \frac{1}{2} - \frac{\lambda }{2} e^{-2 \xi _{*}} $, then $\epsilon$ vanishes faster than $x$ as $x \rightarrow 0$, leading again to $T(\xi_*) = \infty$, and a contradiction.

Hence, there is no $\xi_*$ such that $u(\xi_*) = 0$. 
\end{proof}

\begin{lemma}
If $\lambda > 0$,   then $u(\xi)<1$   for all $\xi \ge 0 $.
\end{lemma}
 \begin{proof}
We integrate Eq. (\ref{seteq}) from the boundary inwards. Let as assume that $u$ first vanishes at $\xi = \xi_c$. 
 Defining $\epsilon := 1-u >0$ and $x := \xi^* - \xi > 0$, we find that in a small neighbourhood of $\xi = \xi_*$, 
 \bey
 \frac{d \epsilon}{d x} =  1 - 2\tilde{v}_* \\
 \frac{dt}{dx} = -  \frac{  \frac{1}{2} + \tilde{w}_* }
 {\epsilon},
 \eey
 where $\tilde{v}_* = \tilde{v}(\xi_*)$, and $w_* = \tilde{w}(\xi_*)$. 
 
 To reach $\epsilon = 0 $, it is necessary that 
  $v_* \leq \frac{1}{2}$. The argument then proceeds as in Lemma 1, the key difference being that we employ the condition $\tilde{w}_* \leq \tilde{v}_*$, which is guaranteed to hold for $\Lambda > 0$. Again we obtain a divergence of $T(\xi_*)$ which contradicts the bound on $v_*$.
 
  Hence, there is no $\xi_*$ such that $u(\xi_*) = 0$. 
\end{proof}

A weaker version of Lemma 1 and 2 was first proved by Winter \cite{Winter}, generalizing previous work for $\Lambda = 0$  \cite{ST}. Winter proved that the solution is non-singular, and develops no horizons, at least up to a point where the mass function $m(r)$ becomes negative. This analysis does not impose thermodynamic constraints on the EoS, and for this reason there is no guarantee that the solution will not encounter a singular behavior before the center (in the region where  $m(r)<0$)---such singularities being artifacts of the unphysical EoS \cite{AnSav21}. In contrast, Lemma 1 and Lemma 2 assume thermodynamic consistency, and for this reason they apply to the full interior of the solution.

\medskip

Since $u$ does not become equal to unity when integrating inwards, and assuming that the pressure $P$ and the density $\rho$ are differentiable functions of temperature, the Picard-Lindel\"of theorem implies that the system (\ref{seteq}) can be integrated to all values of $\xi$, or, equivalently, it can integrated arbitrarily close to $r = 0$. This means that $\lim_{r \rightarrow 0} m(r) \leq 0$, otherwise a horizon would appear near $r = 0$. Since we assume $m(r_B) = M > 0$, we have an obvious corollary.

\begin{corollary}
There is a point $r_1 < r_B$, such that $m(r_1) = 0$. If $r_1 > 0$, then $T(r_1)$ is always finite.
\end{corollary}

If $r_1 = 0$, the solution is regular: there is no singularity at the center.
Since $\{0\}$ is a subset of measure zero in the interval $[0, r_B]$, regular solutions define a set of measure zero in the space of all solutions. If we parameterize the latter by the boundary values $_0 = t(0)$ and $u_0 = u(0)$, regular solutions define a curve in the $u_0-t_0$ plane---see Sec. 5.2 for a specific example.

Typical regular solution have finite $T(0) = T_0$, and their behavior near the center is the following:
\bey
m(r) &=& \frac{4\pi}{3} \rho(T_0) r^3 \\
T(r) &=& T_0 \left[1 -\frac{2\pi}{3}(\rho(T_0)+3P(T_0) - \Lambda)r^2\right]
\eey
We note that temperature increases near the center, except possibly for solutions at large positive $\Lambda$. 

There are, however, regular solutions with divergent temperature. Such solutions define a set of measure zero in the space of all regular solutions. Their form depends on the asymptotic form of the EoS at high temperature. For the expected asymptotic EoS $P = \frac{1}{3}\rho = \sigma T^4$, there is a unique solution of this form, with $u(0) = \frac{1}{7}$, $m(r) \simeq \frac{1}{14}r$, and $T(r) = C/\sqrt{r}$ for some constant $C > 0$.

If $r_1 > 0$, the solution is singular. In such solutions, the contribution from the cosmological constant is suppressed near $r = 0$, so the behavior of singular solutions   is the same as the asymptotically flat case. There are no solutions with $m(0) \rightarrow -\infty$, because for such solutions the derivative $m'(r)$ cannot remain positive in the limit $r \rightarrow 0$, as required by Eq. (\ref{eq5}) and the positivity of the energy density $\rho$. (This point was missed in Ref. \cite{AnSav21}, which excluded such solutions by a mild assumption on the EoS.) Hence, all singular solutions have $m(0) = - M_0$, for $M_0 > 0$. Furthermore:

\begin{lemma}
In a singular solution, $T(r) \sim \sqrt{r}$ as $r\rightarrow 0$. 
\end{lemma}
\begin{proof}
In a singular solution with $m(0) = - M_0$, Eq. (\ref{tovt}) becomes near $T = 0$,
\bey
\frac{d\ln T}{dr} = \frac{1}{2r} - \frac{4\pi r^2 P}{M_0}.
\eey
To prove the result it suffices to demonstrate that $P$ cannot diverge with $r^{-3}$ or faster. If this were the case, $\rho \geq P$ would also diverge with $r^{-3}$ or faster. By Eq. 
(\ref{eq5}), $m(r)$ would diverge  at $r = 0$, contradicting the hypothesis.
\end{proof}
With Lemma 3, the proof of Theorem 1 is concluded.

\subsubsection*{Properties of the singularity}

Since $\rho(T)$ grows at least with $T^2$ near $T = 0$, $m(r) +  M_0$ grows at most with $r^4$ near the center---and much more slowly for realistic equations of state. The pressure term in the TOV-$\Lambda$ equation is much smaller than the cosmological constant term, so near the center the spacetime corresponds to a vacuum solution. The metric is well appoximated by the (anti)deSitter-Schwarzschild solution with negative mass $-M_0$,
\bey
ds^2 = - \Bigl( 1 + \frac{2M_0}{r} -\frac{\Lambda}{3}r^2 \Bigr)d\tilde{t}^2 + \Bigl( 1 + \frac{2M_0}{r} -\frac{\Lambda}{3}r^2 \Bigr) ^ {-1}   dr^2 +r^2 d \Omega ^2.
\eey
where $\tilde{t} = a t$ is an appropriately rescaled time coordinate.

The behavior of the singularity at $r = 0$ is $\Lambda$-independent, and it has been analyzed in Ref. \cite{AnSav21}. It is a curvature singularity, but rather mild as it leaves the spacetime   bounded-acceleration complete: no causal curve with finite acceleration terminates at the singularity. In fact the spacetime is conformal to an ultra-static spacetime with a boundary corresponding to the surface $ r = 0$.

Near the singularity, $dT/dr > 0$, or equivalently $t' < 0$. The decrease of temperature as we integrate inwards is an unfamiliar feature in stellar interiors modeled by the TOV equation. Indeed, for $\Lambda = 0$, $t'(\xi)$ is non-negative in regular solutions. However, regions with  $t' < 0$ always exist in singular solutions---and they are even possible in regular solutions if $\Lambda > 0$.

The regions with negative $t'$ are constrained by the following lemma.

\begin{lemma}
If $u'(\xi) > 0$, then $t'(\xi) > 0$.
\end{lemma}
\begin{proof}
If $u'(\xi) > 0$, then $u > 2 \tilde{v}$. Hence, $t'(\xi) > (\tilde{v} + \tilde{u})/(1-u) =  (v+w)/(1-u) > 0$.
\end{proof}
By Theorem 1, no solution crosses $u = 1$, so a solution  with $u'(0) > 0$ will reach a maximum value $u_{max} \in (0, 1)$ at a point $\xi_{max}$, and then have $u'(\xi) < 0$ and $t'(\xi) > 0$ for some interval $(\xi_{max}, \xi_*)$. This means that for the analysis of temperature inside the solution, we can restrict our considerations to solutions with $u'(0) < 0$, without loss of generality.

For $\Lambda = 0$, the lack of horizons and the behavior near the singularity suffices to determine the overall behavior of the singular solutions. All such solutions start with negative $dT/dr$ at the boundary, and integration reaches a point $r_2 < r_1$, where $dT/dr = 0$. Further integration connects with the solution at the center, where $T(r) \sim \sqrt{r}$. The presence of the cosmological constant changes this picture. However, the behavior of temperature in the solution depends on the sign of $\Lambda$, and we have to distinguish between the two cases.

\section{Properties of asymptotically anti-de-Sitter solutions}

\subsection{Unique temperature maximum}

For $\Lambda < 0$,  $\frac{dT}{dr} < 0$ when $m(r) > 0$. This implies  that $\frac{dT}{dr} <  0$, for $r > r_1$, where $r_1$ is defined by $m(r_1) = 0$. In regular solutions $r_1 = 0$, hence, $dT/dr < 0$ everywhere. For singular solutions, $dT/dr > 0$ near $r = 0$. Invoking continuity, we obtain the obvious corollary.

\begin{corollary}
In singular solutions with $\Lambda < 0$, there exists $r_2 < r_1$, such that $\frac{dT}{dr}(r_2) = 0$. 
\end{corollary}

It turns out that the temperature maximum at $r_2$ is unique. 

\begin{lemma}
In singular solutions with $\Lambda < 0$, $dT/dr$ vanishes only once. 
\end{lemma}
 
\begin{proof}
The roots of $dT/dr = 0$ coincide with the roots of $\phi(r) = m(r) + 4\pi r^3 P - \frac{1}{3}\Lambda r^3$. 
Assume that there are more than one roots of $dT/dr = 0$. Let us denote the largest by $r_2$, and the next largest by $r_3$. 
Since in $r_2$, $\phi(r)$ passes from positive values to negative, $\phi(r) \leq 0$ for all $r \in [r_3, r_2]$. Hence, $\phi(r)$ has a minimum in $(r_3, r_2)$. 

We compute $\phi' = 4\pi r^2(\rho + 3P) -  \Lambda r^2 + 4\pi r^3 (\rho + P) \frac{d \ln T}{dr}$. By Eq. (\ref{tovt}), $\frac{d \ln T}{dr}$ has the opposite sign of $\phi(r)$, so it is non-negative in $[r_3, r_2]$. But then, $\phi'(r) > 0$ in $[r_3, r_2]$ and there is no minimum. We obtained a contradiction.
\end{proof}

\subsection{Approximate-horizon solutions}
Τhe TOV equations admit approximate-horizon solutions, that is solutions in which $u$ approach very close to unity in the interior, before it starts decreasing and becomes negative. Such solutions are characterized by very small values of the temperature at the boundary. In fact, by making $t(0)$ sufficiently small, we can obtain solutions in which $u$ comes arbitrarily close to unity.
 
In Ref. \cite{AnSav16}, it was shown that these solutions can be used to model a black hole in thermal equilibrium with its Hawking radiation, the region where $u$ is sufficiently close to unity behaving effectively as a genuine horizon due to quantum or thermodynamic fluctuations. This is a strong motivation to understand the behavior of such solutions when $\Lambda \neq 0$.

Consider a solution with $\lambda = \Lambda r_B^2 < 0$, and $\bar{u}_0 = 2M/r_B << 1$. Further assume that the temperature at the boundary $T_B$ is so low that the density $\rho$ is much smaller than $\lambda$. This means that we can integrate the solution up to a point $r_* < r_B$
where the density $\rho$  starts becoming comparable to $|\Lambda|$. In the interval $[r_*, r_B]$, the solution is well approximated by a vacuum solution, so the temperature $T_*$ at $r_*$ satisfies,
\bey
\sqrt{1-u_0} T_0 = T_* \sqrt{1-u_*}, \label{adss}
\eey
where $u_* = u(r_{max})$. Since the order of magnitude of $T_*$ is determined by $\Lambda$ and the EoS, we can always choose $T_B$ sufficiently small so that $\epsilon_* = 1 - u_* << 1$. 

By Eq. (\ref{limit}), $\Lambda r_*^2  = \lambda e^{-2 \xi_*} = \simeq f(\bar{u}_0\sqrt{\lambda})$, where $f(x) = 4 \sinh\left[\frac{1}{3} \sinh^{-1}\left(\frac{3}{2}x\right)\right]$. 
Since $f(x) < x$ everywhere, we can always choose $\bar{u}_B$, so that $\Lambda r_*^2 << 1$. Hence, when integrating from $r_*$ inwards the contribution from $\Lambda$ is negligible.
Write $u = 1 - \epsilon$,   Eqs. (\ref{eqqq}) become approximately
\bey
t' = \frac{ \frac{1}{2} +w}{\epsilon}, \;\;\; \epsilon' = 2v -  1 \label{aad}
\eey
We see that $\epsilon$ decreases down to a minimum value $\epsilon_{c}$ at some point $\xi_c$ where $v = \frac{1}{2}$. The sphere $r = r_c = r_B e^{-\xi_c}$ is the approximate horizon. 
Solving Eq. (\ref{aad}) for a given EoS enables us to connect the parameters of the approximate horizon to the boundary data at $r = r_B$.

For example, for the radiation EoS, $P = \frac{1}{3}\rho$, we obtain 
\bey
\frac{d  v}{d\epsilon} =   \frac{v(1 + \frac{2}{3} v)}{\epsilon( v-\frac{1}{2})}, 
\eey
with solution 
\bey
\epsilon = \epsilon_c \frac{\left(v+\frac{3}{2}\right)^2}{4 \sqrt{2v}}.
\eey
At $r_*$, $v_*<< 1$, so $\epsilon_* = \frac{9}{16} \epsilon_c/\sqrt{2v_*}$. Since for radiation $\rho \sim T^4$, we have $v_0/v_* = (T_0/T_*)^4$. Combining with Eq. (\ref{adss}), we find an approximation for the blueshift $\epsilon_c$ of the approximate horizon,
\bey
\epsilon_c = \frac{16\sqrt{2v_0}}{9}(1 - u_0).
\eey
Indeed, we can make $\epsilon_c$ arbitrarily small  by choosing sufficiently small $v_0$. Approximate-horizon solutions exist for all values of $\Lambda < 0$.

\subsection{Integrating to infinity}
So far, we have focused on the integration of the TOV-$\Lambda$ from a boundary at $r = r_B$ inwards. The primary interpretation is that of a thermodynamic system in equilibrium localized inside the sphere $r = r_B$ connecting with a vacuum solution at $r > r_B$. 

However, nothing forbids integration from the boundary outwards. The issue here is whether this integration can reach infinity, and what type of system  describes. For $\Lambda > 0$, integration cannot proceed beyond the cosmological horizon at a finite value of $r$.

For $\Lambda \leq 0$ the situation is different. It is straightforward to adapt Lemma 1, to show that integration outwards ($\xi < 0$) cannot encounter a singularity. 

However, for $\Lambda = 0$, no solution with finite mass is possible. In a thermodynamically consistent EoS, the energy density only vanishes for $T = 0$, which cannot be reached at finite $\xi$. This means that the integration outwards cannot terminate at a surface with  $\rho = 0$---as is the usual assumption in the study of regular solutions. Finite-mass solutions---relevant for isolated self-gravitating bodies---require a surface $r = r_B$ where outwards integration terminates; the connection to an external vacuum solution for $r > r_B$ necessarily involves a discontinuous first derivative of the metric, essentially implying the presence of a bounding box. 

To see this, assume that   $m(r) \rightarrow  M$ at large $r$.  Then, $d \ln T/dr = - M/r^2$, hence $T(r)$ converges to a constant $T_{\infty}$ as $r \rightarrow \infty$. But then the energy density and pressure at infinity are non-zero, hence, $m(r)$ grows at least with $r^2$, in contradiction to the initial assumption.

For $\Lambda = 0$, vacuum solutions at $r \rightarrow \infty$ have constant $u$, with temperature vanishing asymptotically. Hence, in absence of a bounding box, the mass function   diverges linearly.

In contrast, solutions for $\Lambda < 0$ describe finite self-gravitating bodies even when $r$ is taken to infinity. Since there is no horizon when integrating outwards, $m(r)$ grows more slowly than $r^3$, hence $\rho$ vanishes at infinity. It follows that the cosmological constant term dominates both numerator and denominator in the right-hand side of Eq. (\ref{tovt}) for large $r$, leading to $d \ln T/dr = -r^{-1}$ and  an asymptotic behavior $T \simeq \alpha/r$ for some constant $\alpha$. 

By Eq. (\ref{exponent}), $\rho$ vanishes at least as fast as $\alpha^2/r^2$. The limiting behavior $\rho \sim r^{-2}$ correspond to the EoS $P = \rho$. For this EoS, $dm/dr$ is asymptotically constant, so $m(r)$ grow linearly as $r \rightarrow \infty$. The mass function diverges at infinity for all EoS with $P/\rho \geq \frac{1}{2}$, for which $\rho$ drops more slowly than $T^3$ as $T \rightarrow 0$. 

Such EoS are thermodynamically acceptable, but they are not physical. As explained in Sec. 2.2, QFT suggests that:
\begin{itemize}
\item In a theory with mass gap $\bar{m}$, $\rho \simeq C T^{5/2}e^{-\bar{m}/T}$ for some constant $C$. Hence, $\rho(r) = C \alpha^{5/2} r^{-5/2} e^{-\bar{m}\alpha/r}$, leading to an asymptotic beehavior
    \bey
    m(\infty) - m(r) \sim \sqrt{r} e^{-\bar{m}\alpha/r}.
    \eey

\item In a theory with no mass gap, $\rho \simeq \sigma T^4$ for some constant $\sigma$. Hence, $\rho(r) = \sigma \alpha^4 r^{-4}$, leading to 
\bey
m(\infty) - m(r)   \sim  r^{-1}.
\eey
\end{itemize}

\section{Properties of asymptotically  de-Sitter solutions}

\subsection{General properties}

Solutions to the TOV-$\Lambda$ equation with $\Lambda > 0$ differ from those with $\Lambda \leq 0$, because for $\Lambda > 0$, $t'$ can become negative even for positive $u$. This implies in particular that there are solutions with $t'(0) < 0$.

Since the pressure is an increasing function of temperature, there is a single value $t = t_c$, such that $P(t_c) = \frac{\Lambda}{8\pi}$; for $t > t_c$, $\tilde{P}(t_c) > 0$, and for $t < t_c$, $\tilde{P}(t_c) < 0$.

To classify the behavior of solutions with $\Lambda > 0$, we use two results from Sec. 3:
\begin{itemize}
\item To analyse the behavior of temperature, we can assume without loss of generality that $u'(0) < 0$.
\item The cosmological constant dominates over matter near the singularity, so $\tilde{w}(\xi) \simeq -\frac{1}{2}\lambda e^{-2\xi}$ as $\xi \rightarrow \infty$; that is, $\tilde{w}(\xi) < 0$ and $\tilde{w}'(\xi) > 0$ at large $\xi$.
\end{itemize}
Therefore, we consider all solutions with $u'(0)< 0$ and $u(0) > 0$, starting integration at positive $u$ and moving towards $u\rightarrow -\infty$. We will be using the following equations.
\bey
\tilde{w}' &=& -2 \tilde{w} + (w+v) t' \label{wprime}\\
t'' &=& \frac{u'(\frac{1}{2}+\tilde{w})}{(1-u)^2} + \frac{w'}{1-u}, \label{t2prime}
\eey
straightforwardly obtained from Eq. (\ref{seteq}). 

\subsubsection*{Case 1. ${\bf t'(0) > 0}$, ${\bf \tilde{w}(0) > 0}$}
If the solution has initially $\tilde{w}'(0)$ positive, $\tilde{w}(\xi)$ increases up to a maximum and then starts decreasing. Hence, without loss of generality, we can always choose the starting point of the integration to have $w'(0) < 0$. 

Even if $w$ decreases, $t'$ remains positive, as long as  $u > 0$ and $\tilde{w} > 0$. Hence, the solution curve will cross the $\tilde{w}$-axis at a point $\xi_1$, with $t(\xi_1) > t(0) > t_c$, whence $\tilde{w}(\xi_1) > 0$. No solution curve can cross the $u$-axis before it crosses the $\tilde{w}$-axis. 

After $\xi_1$, the solution curve will cross the $t' = 0$ line, then the $u$-axis, it will reach a minimum for $\tilde{w}$, and  finally it will develop the asymptotic behavior of all solutions as     $u \rightarrow - \infty$ with $\tilde{w}' > 0$.

These solution curves behave similarly to their counterparts for $\Lambda \leq 0$. We will refer to them as {\em Type I} solutions.





\subsubsection*{Case 2. ${\bf t'(0) > 0}$, ${\bf \tilde{w}(0) < 0}$}
By Eq, (\ref{wprime}),  
 $w'(0) > 0$, and  $w'$ remains positive as long as the solution curve lies between the positive  $u$-axis and the $t' = 0$ line. 

 In this region, $t < t_c$ and $t' > 0$, while by Eq. (\ref{t2prime}), $t'' < 0$. Hence, $t(\xi)$ along the solution curve is an increasing and  concave function. If it reaches  $t = t_c$ before encountering a local maximum, then the solution curve will cross the $u$-axis, it will achieve a maximum in $\tilde{w}$, after which it behaves like the solutions of Case 1, i.e., it is of Type I.

However, if $t(\xi)$ achieves its maximum before $t = t_c$, the solution curve crosses the line  $t' = 0$ first. Then, it    continues with $w' > 0$  to its asymptotic regime as $u \rightarrow -\infty$. 

The latter curves, crossing the $t' = 0$ line with $u > 0$ have no counterpart for $\Lambda \leq 0$. We will refer to them as {\em Type II} solutions.

\subsubsection*{Case 3. ${\bf t'(0) < 0}$, ${\bf \tilde{w}(0) < 0}$}
In this case, $\tilde{w}'(0)$ can be either positive or negative. In the first case, Eq. (\ref{t2prime}) implies that the solution curve starts with $t'' >0$, hence, $t'$ increasing. 

If $t'$ does not encounter a maximum before the line $t' = 0$, the solution curve will go on upwards, cross the $u$-axis, and behave as in Case II or Case I afterwards. Overall, the curve crosses the line   $t' = 0$ twice. We refer to these solutions as {\em Type III}.

However, if $t'$ reaches a maximum before $t' = 0$,  the solution immediately transitions into the asymptotic behavior as $u \rightarrow -\infty$.   It never crosses the line $t' = 0$. These are {\em Type IV} solutions.

There is no significant difference if $\tilde{w}'(0) < 0$. The solution reaches a minimum value of $\tilde{w}'(0)$ and then, as $\tilde{w}'(\xi) > 0$, it behaves either as a Type III or a Type IV solution. 

\begin{figure}[]
\includegraphics[height=3.7cm]{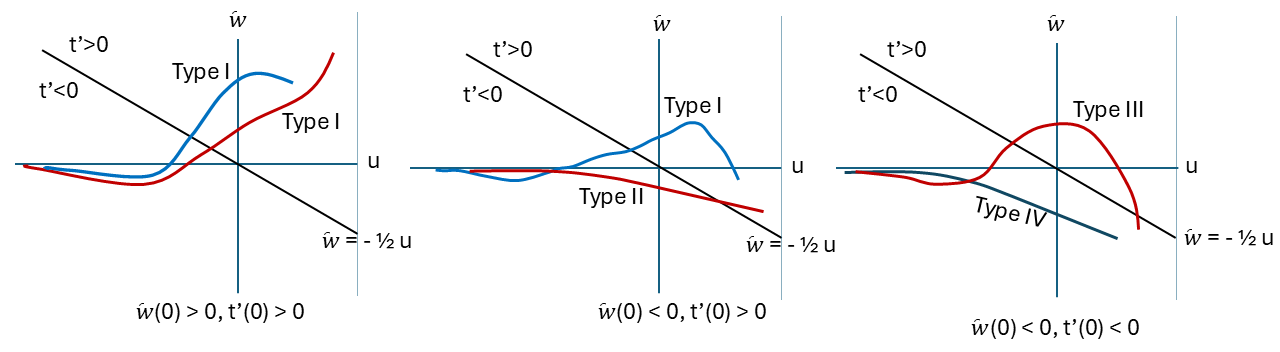} \caption{ Solutions to the TOV-$\Lambda$ equation for $\Lambda > 0$. For $t'(0) >0$ and  $\tilde{w}(0) > 0$, there are only Type I solutions. For For $t'(0) >0$ and  $\tilde{w}(0) < 0$, there are both Type I and Type II solutions.
For $t'(0) <0$ and  $\tilde{w}(0) < 0$, there  two types of solution. Type III solutions cross the $t' = 0$ line twice, while type-IV solutions never cross it.  
}\label{uw2}
\end{figure}

\subsubsection*{Regular solutions}
As shown in Sec. 3, regular solutions have $u > 0$ everywhere, and they terminate at the coordinate origin $u = \tilde{w} = 0$. They may reach the origin from below or from above the $u$-axis, as $\tilde{P}$ can be either positive or negative in the vicinity of the center. 

Regular solutions are limiting cases of the singular ones, where the intercept of the solution curve with the $\tilde{w}$-axis occurs at $\tilde{w} = 0$. For small $\lambda$, they are limiting cases of Type I curves, but as $\lambda$ increases they may arise as limits of any of the four types of solution.

\subsection{Example: Linear equation of state}
We will consider the EoS for radiation $P = \frac{1}{3} \rho$. In this case, it is convenient to express the TOV-$\Lambda$ equation in terms of $\tilde{v}$, 
\bey
v' = 2v \frac{ 1 - 2u -\frac{2}{3} v -  \lambda   e^{-2 \xi}     }{1 - u},
\eey
We characterize the solutions by their initial values $u_0 = u(0)$ and $v_0 = v(0)$. Since we only consider solutions of positive mass $M$, $u_0 > \frac{1}{3}\lambda$.

Indicative plots of $\tilde{m}(r)$ and $T(r)$ for the four types of solution are plotted in Fig. (\ref{uw3}).

\begin{figure}[]
\centering

\includegraphics[width=\linewidth]{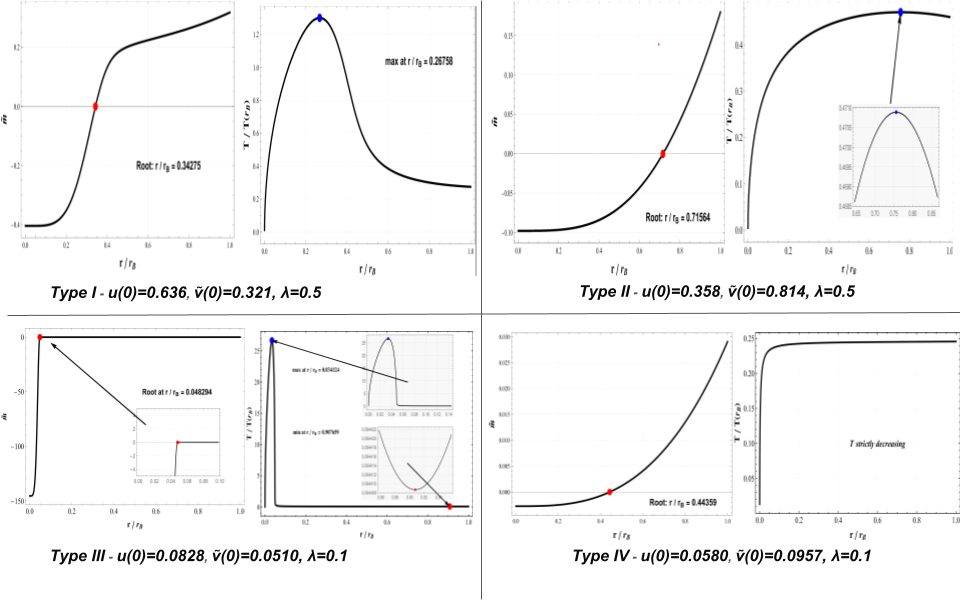}
\caption{Indicative plots of the effective mass function $\tilde{m}$ and the temperature $T$ as a function of $r/r_B$ for the four solution types. } 
\label{uw3}
\end{figure}

Fig. \ref{fig:classification} shows how the space of initial conditions is partitioned among the four types of solution.
 For small values of $\lambda$, Type I solutions cover almost the full allowed range, with the other three types of solution restricted to a small corner. The area corresponding to Type II, III, and IV solutions increases significantly with $\lambda$.

The curve $\tilde{w}(0) = 0 $  corresponds to the line $v_0 = \frac{3}{2} \lambda$. As expected by our previous analysis, there are only type I solutions above that line. 

  The curve $t'(0) = 0$ corresponds to the straight line $u_0 + \frac{2}{3} v_0 = \lambda$. Solutions of Type III and Type IV exist only below that line. 

  The curve $u'(0) = 0$ corresponds to the straight line $u_0 - 2 v_0 = \lambda$. For solutions beneath that line $u$ initially increases inwards until it reaches a maximum, and then starts decreasing. They are all of Type I. Note that the minimum $u_0$ value for such solutions is $\lambda$, so there is no possibility for approximate horizon solutions unless $\lambda << 1$. For $\lambda > 1$, there are no solutions with $u'(0)>0$, the dividing line lies wholly beneath the horizontal axis.

\begin{figure}[h]
    \centering

    \begin{subfigure}{0.45\textwidth}
        \centering
        \includegraphics[width=\linewidth]{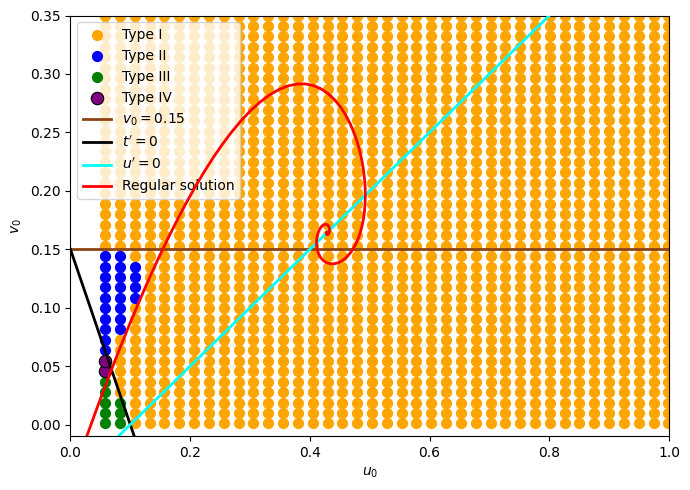}
        \caption{$\lambda = 0.1$}
    \end{subfigure}
    \hfill
    \begin{subfigure}{0.45\textwidth}
        \centering
        \includegraphics[width=\linewidth]{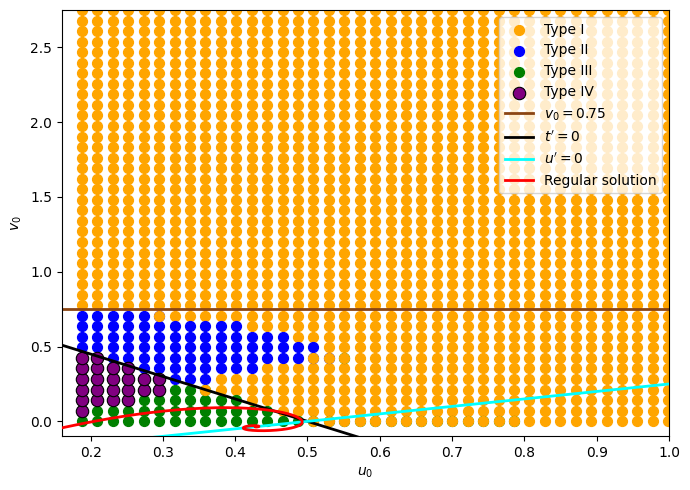}
        \caption{$\lambda = 0.5$}
    \end{subfigure}
    \hfill
    \begin{subfigure}{0.45\textwidth}
        \centering
        \includegraphics[width=\linewidth]{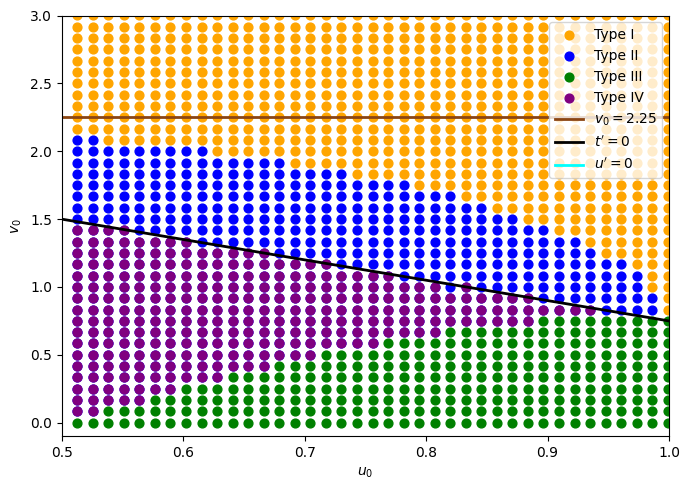}
        \caption{$\lambda = 1.5$}
    \end{subfigure}

    \caption{ The physically relevant   initial conditions to the TOV-$\Lambda$ equation for $\Lambda > 0$ correspond to pairs $(u_0, v_0)$ with $u_0 \in [\frac{\lambda}{3},1]$ and $v_0> 0$. This set is partitioned into regions that  correspond to the four different types of singular solutions. The spiraling curves in (a) and (b) correspond to regular solutions, defining a set of measure zero in the space of solutions. There are no regular solutions with $M > 0$ for $\lambda = 1.5$.
    }
    \label{fig:classification}
\end{figure}

\section{Conclusions and outlook}

In this work, we have provided a general classification of solutions to the TOV equation in the presence of a cosmological constant, for arbitrary thermodynamically consistent equations of state. By formulating the problem as an initial value system and integrating from an outer boundary inwards, we identified the generic structure of solutions and showed that singular configurations dominate the solution space. Despite their singular character, these solutions exhibit a universal geometric structure and lead to spacetimes that are bounded-acceleration complete, indicating that the associated singularities are comparatively mild.

The present results extend and unify our earlier work on singular TOV solutions to the case $\Lambda \neq 0$. We show that the structural features identified for $\Lambda=0$ persist in a modified form, while new phenomena arise due to the cosmological constant. For $\Lambda < 0$, we identified solutions with approximate horizon structures that mimic black holes in equilibrium with their Hawking radiation. For $\Lambda > 0$, we found four distinct classes of solutions, reflecting the interplay between gravitational attraction, cosmic repulsion, and cosmological horizons. The structure of the different types of solution are summarized in Fig. (\ref{uk2}).

\begin{figure}[h]
\includegraphics[height=4cm]{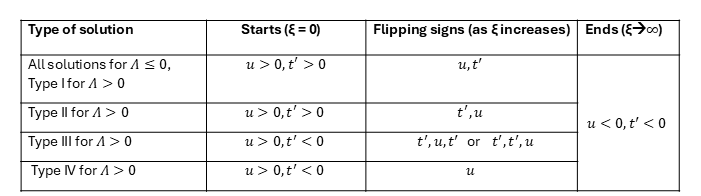} \caption{ Classification of singular solutions to the TOV-$\Lambda$ equation. All solution are assumed to have $u' <0$ at the outer boundary. No loss of generality is involved as solutions with $u' >0$ initially reach a maximum value of $u$ and then continue with $u' <0$.
}\label{uk2}
\end{figure}

Several directions for further investigation follow. A key open question is the dynamical origin of the singular solutions classified here, in particular whether they can arise as endpoints or intermediate stages of gravitational collapse. The thermodynamic role of singular solutions also deserves further study, especially in relation to entropy assignment and its extension to $\Lambda \neq 0$ backgrounds, where cosmological or AdS boundary conditions introduce additional structure.

Our results also raise questions about physical interpretation in different asymptotic regimes. For $\Lambda > 0$, the distinct solution classes may model high-density configurations in a cosmological background, potentially relevant to early-universe physics. For $\Lambda < 0$, the existence of horizon-like solutions suggests a holographic interpretation, where bulk singularities encode properties of the dual strongly coupled field theory.

Finally, it would be important to extend the present analysis beyond static, spherically symmetric configurations and determine the robustness of the structures identified here. In particular, one should investigate whether the singular solutions are stable under perturbations, and whether analogous structures persist in static spacetimes with a different isometry group.

\section*{Acknowledgments}
CD acknowledges support  by a grant from the Andreas Mentzelopoulos Foundation.

\end{document}